\documentclass[a4]{article}
\usepackage{graphicx,graphics}%,here
\usepackage{theorem}
 \usepackage{amsmath} \usepackage{amssymb,amsfonts}%,textcomp}
\usepackage{stmaryrd}
\newcommand{\qed}{\hfill QED.}
\newtheorem{theorem}{Theorem}
\newtheorem{lemma}[theorem]{Lemma}

\newtheorem{observation}[theorem]{Observation}
\newtheorem{definition}{Definition}

\newenvironment{proof}{\textbf{Proof.}}{\medskip}
\newcommand{\LV}[1]{#1}
\newcommand{\SV}[1]{}

\begin{document}
\SV{\mainmatter}
\SV{\title{Approximation Algorithms Inspired by Kernelization Methods}}
\LV{\title{Data Reductions and Combinatorial Bounds for Improved Approximation Algorithms\\
% \tnoteref{note1}}
% \tnotetext[note1]{
{\small This work is supported by the bilateral research cooperation
CEDRE between France and  Lebanon (grant number
30885TM). An extended abstract will be presented at ISAAC 2014.}}}
%OR: Approximation-Preserving Reductions for Duals of Domination Problems
\SV{\titlerunning{Approximation Algorithms Inspired by Kernelization Methods}}
%\LV{\titlerunning{Data Reductions %for Approximation}}

\SV{\maketitle}

\thispagestyle{empty}
\begin{center}
\textbf{\Large Data Reductions and Combinatorial Bounds for Improved Approximation Algorithms}\footnote{This work is supported by the bilateral research cooperation
CEDRE between France and  Lebanon (grant number
30885TM). An extended abstract will be presented at ISAAC 2014.}\\[2ex]
% \LV{
% \author[beirut]
{Faisal N. Abu-Khzam}
% \ead
\\
{Lebanese American University, Beirut, Lebanon}\\
{faisal.abukhzam@lau.edu.lb}
\\[1ex]
% \author[paris]
{Cristina Bazgan}
% \ead
\\
{PSL, University of Paris-Dauphine, LAMSADE UMR 7243, France;\\ and Institut Universitaire de France}
\\
{bazgan@lamsade.dauphine.fr}
\\[1ex]
% \author[ulm]
{Morgan 
Chopin}
\\
{Institut f\"ur Optimierung und Operations Research, Universit\"at Ulm, Germany}
% \ead
\\
{morgan.chopin@uni-ulm.de}
\\[1ex]
% \author[trier]
{Henning Fernau}%\corref{cor1}}
% \ead
\\{Fachbereich 4, Abteilung Informatikwissenschaften, 
Universit\"at Trier,  %D-54286 Trier, 
Germany}\\
{fernau@uni-trier.de}
% \cortext[cor1]{Corresponding author; Tel. +49 651 201 2827.}
% 
% 
% \address[beirut]{Lebanese American University, Beirut, Lebanon}
% 
% \address[paris]{PSL, University of Paris-Dauphine, LAMSADE UMR 7243, France; and \\Institut Universitaire de France}
% 
% \address[ulm]{Institut f\"ur Optimierung und Operations Research, Universit\"at Ulm, Germany}
% 
% \address[trier]{Fachbereich 4, Abteilung Informatikwissenschaften, 
% Universit\"at Trier,  %D-54286 Trier, 
% Germany}
%}
\end{center}
\SV{
\author{Faisal N. Abu-Khzam\inst1, Cristina Bazgan\inst{2,5}, Morgan 
Chopin\inst3, Henning Fernau\inst4}
\authorrunning{Abu-Khzam, Bazgan, Chopin and Fernau}
\institute{Lebanese American University, Beirut, Lebanon\\
\email{faisal.abukhzam@lau.edu.lb}
\and
PSL, University of Paris-Dauphine, LAMSADE UMR 7243, France\\
\email{bazgan@lamsade.dauphine.fr}
\and
Institut f\"ur Optimierung und Operations Research, Universit\"at Ulm, Germany\\
\email{morgan.chopin@uni-ulm.de}
\and Fachbereich 4, 
Informatikwissenschaften, 
Universit\"at Trier, Germany\\
\email{fernau@uni-trier.de}
\and
Institut Universitaire de France}
}

\SV{\maketitle}

\thispagestyle{empty}

\begin{abstract}
Kernelization algorithms in the context of Parameterized Complexity  are often based on a combination of reduction rules and combinatorial insights. We will expose in this paper
a similar strategy for obtaining polynomial-time approximation algorithms.
Our method features the use of approximation-preserving reductions, 
%by employing a relaxed version of 
akin to the notion of parameterized reductions.
%, 
%often used in Kernelization algorithms.
We exemplify this method to obtain the currently best approximation algorithms for
\textsc{Harmless Set}, \textsc{Differential} and \textsc{Multiple Nonblocker},
all of them can be considered in the context of securing networks or information propagation.
%The Harmless Set problem takes as input a graph $G$, together with an 
%integer-valued function $\threshold$ on its vertices, 
%and asks for a largest set of vertices $S$
%such that every vertex of $G$ has less than $\threshold(v)$ neighbors
%in $S$.
%When $\threshold$ is the vertex-degree function,
%the complement of a harmless set is a total dominating set 
%and corresponding problem, dubbed Unanimous Harmless Set,
%has a factor-three polynomial-time approximation algorithm. Using
%known bounds on the total dominating set of a graph,
%we obtain a factor-two polynomial-time approximation algorithm
%for Unanimous Harmless Set.
%%In addition to the result pertaining to Unanimous Harmless Set,
\end{abstract}

% \LV{
% \begin{keyword} 
\textbf{Keywords}: Reduction rules, maximization problems, 
 polynomial-time approximation, domination problems
% \end{keyword}
% \maketitle}

\section{Introduction}

It is well-known that most interesting combinatorial problems are hard from a computational point of view.
More technically speaking, they mostly turn out to be NP-hard. 
As many of these combinatorial problems have some importance for practical applications,
several techniques have been developed to deal with them.
From a more mathematical angle, the two most interesting and wide-spread approaches are (polynomial-time)
approximation and fixed-parameter algorithms.
Both areas have developed their own set of tools over the years.
For instance, methods related to Linear Programming are prominent in the area of Approximation Algorithms~\cite{Ausetal99}.
Conversely, data reduction rules are the method of choice to obtain kernelization results, 
which is central to Parameterized Algorithms~\cite{DowFel2013}. Another essential ingredient to kernelization algorithms
is a collection of combinatorial insights to the specific problem, often (already) supplied by
mathematicians working in Combinatorics.
It is quite natural to try to employ certain tools from one area to the other one.
For example, the title of the paper~\cite{Naretal2012} nicely indicates the intended use of Linear Programming to obtain FPT algorithms. 
In this paper, we take the opposite approach and show how to use data reduction rules and 
(constructive) combinatorial insights to obtain approximation algorithms, in particular for maximization problems.
Notice that data reduction rules are often used in heuristic approaches, well-established in practical implementations. So, our approach also brings the often more theoretical findings closer to practice.

For the purpose of illustrating our method, we will mainly deal with maximization problems that are obtained from domination-type graph problems. We  first describe these problems, using 
standard graph-theoretic terminology.

Let $G=(V,E)$ be an undirected graph and $D\subseteq V$.
\begin{enumerate}
\item $D$ is called a \emph{dominating set} if, for all $x\in V\setminus D$, there is a $y\in D\cap N(x)$.
$V\setminus D$ is known as an \emph{enclaveless set} \cite{Sla77} or 
as a \emph{nonblocker set} \cite{Dehetal2006}.
\item $D$ is called a \emph{total dominating set} if, for all $x\in V$, there is a $y\in D\cap N(x)$.
$V\setminus D$ has been introduced as a \emph{harmless set}
or \emph{robust set} (with unaminity thresholds) in~\cite{BazCho2012}.
\item If $D$ can be partitioned as $D=D_1\cup D_2$ such that, for all $x\in V\setminus D$, there is a $y\in D_2\cap N(x)$, then $(D_2,D_1)$ defines a \emph{Roman domination function} $f_{D_1,D_2}:V\to\{0,1,2\}$ such that $f_{D_1,D_2}(V)=2|D_2|+|D_1|$. According to~\cite{BerFerSig2014}, $D_0:=V\setminus (D_1\cup D_2)$ is also known as the \emph{differential (set)} of a graph (as introduced in~\cite{Masetal2006}) if $f_{D_1,D_2}(V)$ is smallest possible.
\item If for all $x\in V\setminus D$, there are $k$ elements in $D\cap N(x)$, then $D$ is a \emph{$k$-dominating set}, see \cite{CarRod90,CocGamShe85,FinJac85}.
We will call $V\setminus D$ a \emph{$k$-nonblocker set}.
%; if $k=2$, we refer to this as a \emph{double nonblocker set}.
\end{enumerate}

The maximization problems derived from these four definitions are: \textsc{Nonblocker}, \textsc{Harmless Set}, \textsc{Differential}, and \textsc{$k$-Nonblocker}.
\LV{Actually, \textsc{Nonblocker} has been looked into by the approximation algorithm community quite a lot in recent years~\cite{AthCKK2009,CheENRRS2013,NguSHSMZ2008}, where it is known as the \textsc{Maximum Star Forest} problem.}\LV{
}
Although these problems are all better known from the minimization perspective, there is a good reason to study them in this complementary way: All of these minimization problems do not possess constant-factor approximations under reasonable complexity assumptions (the reduction shown in~\cite{ChlChl2008a} for \textsc{(Total) Dominating Set} starts from \textsc{Set Cover}), while the complementary problems can be treated in this favorable way. For \textsc{Roman Domination}, observe that the reduction shown in \cite{Fer08} works from \textsc{Set Cover}, so that again (basically) the same lower bounds follow.
This move is related  to 
%the proposal of considering 
\emph{differential approximation}~\cite{Ausetal2005}.\LV{
 }
Notice that this comes along with similar properties from the perspective of Parameterized Complexity:
While natural parameterizations of the minimizations lead to W[2]-hard problems~\cite{DowFel2013,Fer08}, the natural parameterizations of the maximization counterparts are fixed-parameter tractable.\LV{
}
However, as this is more customary as a combinatorial entity, let us refer (as usual) by $\gamma(G)$ to the size of the smallest dominating set of $G$, by $\gamma_t(G)$ to the size of the smallest total dominating set,
by $\gamma_R(G)$ to the Roman domination number of $G$,
i.e., the smallest value of a Roman domination function of $G$, and by $\gamma_k(G)$ to the size of the smallest $k$-dominating set of $G$.

\paragraph{Some graph-theoretic notations}  Let $G=(V,E)$ be a simple undirected graph. We denote by $N(x)$ the set of neighbors of vertex $x$; the cardinality of $N(x)$ is the \emph{degree} of $x$. A vertex of degree zero is known as an \emph{isolated vertex}, and a vertex of degree one as a \emph{leaf}. The number of vertices of a graph is called its \emph{order}.
Given  $U\subseteq V$, $G[U]$ denotes the subgraph induced by $U$. A repetition-free sequence $x_1,\dots,x_k$ of vertices is a \emph{path} in $G$ (of \emph{length} $k-1$) if $x_ix_{i+1}\in E$ for $i=1,\dots,k-1$.  
A {\em chain} is an induced path whose interior 
vertices are of degree two in $G$. The \emph{diameter} of $G$ is the greatest length of a shortest path in~$G$.

\paragraph{Main Results.}  We introduce a notion
of approximation-preserving reductions analogous to parameter-preserving reductions known in Parameterized Complexity in order to obtain new approximation algorithms. We introduce a general methodology to obtain constant-factor approximations for various problems. 
For instance, along with an algorithmic version of the upper bound 
obtained in \cite{LamWei2007} on the size of a total dominating set,
we present a factor-two approximation algorithm for  \textsc{Harmless Set}, beating the previously known factor of three~\cite{BazCho2012}. 
%We are going to devise reduction rules to exploit this %relationship and 
%also known bounds on the size of total dominating sets %of a certain minimum 
%degree~\cite{HenYeo2007} just as this was done before %for 
%the \textsc{Nonblocker} problem~\cite{Dehetal2006}.
Moreover, we are deriving a  factor-$\frac{11}{3}$  approximation algorithm for \textsc{Differential}, which was set up as an open problem in \cite{BerFer2014}, where this approximability question could be only settled for bounded-degree graphs; our approach also 
improves on the factor-4 approximation exhibited in~\cite{BerFer2013sub}. \LV{However, as  in \cite{BerFer2014} APX-completeness was shown for the degree-bounded case, nothing better than constant-factor approximations can be expected for general graphs. }Finally, we present constant-factor approximation algorithms for $k$-\textsc{Nonblocker}.

\paragraph{Organization of the paper} 
%The rest of this paper is organized as follows.
Section~\ref{sec:max} explains the use of reduction rules within maximization problems.
It also exhibits the general method.
Section\LV{s}~\ref{sec:Harmless}\LV{ and \ref{sec:Harmless-Reductions}} show\SV{s} how to employ our general method to one specific problem in a non-trivial way.
Sections~\ref{sec:Differential} and~\ref{sec:Multiple} show that the same method can be also applied to other problems.
We conclude with discussing further research directions.

%Section 2 presents a few
%background information, followed by a section on our %general approximation
%via reduction framework. Section 4 presents our main 
%approximation, via dual-problem 
%reduction, for Unanimous Harmless Set.  

\SV{All proofs and some more details that are omitted due to space restrictions can be found in the long version of this paper~\cite{AbuBCF2014}.}

\section{Approximation preserving reductions for maximization problems}
\label{sec:max}

%Recall that a 
Specializing standard terminology from~\cite{Ausetal99}, we can express the following. 
A maximization problem $\mathcal{P}$ can be specified by a 
triple $(I_\mathcal{P},\mathrm{SOL}_\mathcal{P},m_\mathcal{P})$, where
\begin{enumerate}
 \item $I_\mathcal{P}$ is the set of input instances of $\mathcal{P}$;
 \item $\mathrm{SOL}_\mathcal{P}$ is a function that associates to $x\in I_\mathcal{P}$ the set $\mathrm{SOL}_\mathcal{P}(x)$ of feasible solutions of $x$;
 \item $m_\mathcal{P}$ provides on $(x,y)$, where $x\in I_\mathcal{P}$ and $y\in\mathrm{SOL}_\mathcal{P}(x)$, a positive integer which is the value of the solution $y$.
\end{enumerate}

An optimum solution $y^*$ to $x$ satisfies: (i) $y^*\in \mathrm{SOL}_\mathcal{P}(x)$, and (ii) $m_\mathcal{P}(y^*)=\max\{m_\mathcal{P}(y)\mid y\in\mathrm{SOL}_\mathcal{P}(x)\}$.
The value $m_\mathcal{P}(y^*)$ is also referred to as $m^*_\mathcal{P}(x)$ for brevity.

Given a maximization problem  $\mathcal{P}$, a \emph{factor-$\alpha$ approximation}, $\alpha\ge 1$, associates to each $x\in I_\mathcal{P}$ some $y\in\mathrm{SOL}_\mathcal{P}(x)$ such that 
$\alpha\cdot m_\mathcal{P}(x,y)\geq m^*_\mathcal{P}(x)$. A solution  $y\in\mathrm{SOL}_\mathcal{P}(x)$ satisfying $\alpha\cdot m_\mathcal{P}(x,y)\geq m^*_\mathcal{P}(x)$ 
is also called an \emph{$\alpha$-approximate solution} for $x$.

We are now going to present a first key notion for this paper.

\begin{definition}
An \emph{$\alpha$-preserving reduction}, with $\alpha\ge 1$, is a pair of mappings $\mathop{\tt inst}_\mathcal{P}:I_\mathcal{P}\to I_\mathcal{P}$ and $\mathop{\tt sol}_\mathcal{P}$ which, given $y'\in \mathrm{SOL}_\mathcal{P}(\mathop{\tt inst}_\mathcal{P}(x))$,
produces some $y\in \mathrm{SOL}_\mathcal{P}(x)$ such that there are constants $a,b\geq 0$ satistying $a\leq \alpha\cdot  b$ and the following inequalities:
\begin{enumerate}
 \item $m^*_\mathcal{P}(\mathop{\tt inst}_\mathcal{P}(x)) + a \geq  m^*_\mathcal{P}(x)$,
 \item for each $y'\in \mathrm{SOL}_\mathcal{P}(\mathop{\tt inst}_\mathcal{P}(x))$, the corresponding solution $y=\mathop{\tt sol}_\mathcal{P}(y')$ satisfies:
 $m_\mathcal{P}(\mathop{\tt inst}_\mathcal{P}(x),y')+b\leq m_\mathcal{P}(x,y)$.
\end{enumerate}
\end{definition}

When referring to this definition, we mostly explicitly specify the constants $a$ and $b$ for ease of verification.
An important trivial example is given by a pair of identity mappings that are $\alpha$-preserving for any $\alpha\geq 1$.
Notice that a similar notion has been introduced, or implicitly used, in the context of  minimization problems in~\cite{BraFer1112,BraFer2013,FomLMS2012}.
%Moreover, an implicit use of $\alpha$-preserving reductions for minimization problems appeared in~\cite{FomLMS2012}.

\begin{theorem}\label{thm-alphaapprox-general}
 Let  $\mathcal{P}=(I_\mathcal{P},\mathrm{SOL}_\mathcal{P},m_\mathcal{P})$ be some maximization problem. If the pair $(\mathop{\tt inst}_\mathcal{P},\mathop{\tt sol}_\mathcal{P})$ describes an $\alpha$-preserving reduction
 and if, given some instance $x$, $y'\in \mathrm{SOL}_\mathcal{P}(\mathop{\tt inst}_\mathcal{P}(x))$ is  an $\alpha$-approximate solution for $\mathop{\tt inst}_\mathcal{P}(x)$, then 
 $y=\mathop{\tt sol}_\mathcal{P}(y')$  is an  $\alpha$-approximate solution for $x$.
\end{theorem}

\begin{proof}
We have to prove that $\alpha\cdot m_\mathcal{P}(x,y)\geq m^*_\mathcal{P}(x)$. Now,
$$\frac{m^*_\mathcal{P}(x)}{m_\mathcal{P}(x,y)}\leq \frac{m^*_\mathcal{P}(\mathop{\tt inst}_\mathcal{P}(x)) + a}{m_\mathcal{P}(\mathop{\tt inst}_\mathcal{P}(x),y')+b}\leq 
\frac{\alpha m_\mathcal{P}(\mathop{\tt inst}_\mathcal{P}(x),y')+\alpha b}{m_\mathcal{P}(\mathop{\tt inst}_\mathcal{P}(x),y')+b}
%= \frac{\alpha(|V_{appr}'|+b)}{|V_{appr}'|+b}
=\alpha$$
as required.\qed
\end{proof}

This shows that an $\alpha$-preserving reduction leads to a special AP-reduction as defined in \cite{Ausetal99}.
But there, these reductions were mainly used to prove hardness results, as it is also the case of \cite{FomLMS2012} that we already mentioned.
However, we use this notion to obtain approximation algorithms.

The notion of an $\alpha$-preserving reduction was coined following the successful example of kernelization reductions known from Parameterized Complexity~\cite{DowFel2013}. One of the nice features of those is that they are usually compiled from simpler rules that are often based on some applicability conditions. In the following, we describe that this also works out for approximation.
We need two further notions to make this precise.

We call an $\alpha$-preserving reduction $(\mathop{\tt inst}_\mathcal{P},\mathop{\tt sol}_\mathcal{P})$ \emph{strict} if $|\mathop{\tt inst}_\mathcal{P}(x)|<|x|$ for all $x\in I_\mathcal{P}$,
and it is called \emph{polynomial-time computable} if the two mappings comprising the reduction can be computed in polynomial time.

The following lemma is relatively straightforward to prove. Yet, it contains an important
message: reduction rules can be composed so that the composition 

\begin{lemma}\label{lem-compose}
 If $(\mathop{\tt inst}_\mathcal{P},\mathop{\tt sol}_\mathcal{P})$ and $(\mathop{\tt inst}'_\mathcal{P},\mathop{\tt sol}'_\mathcal{P})$ are two $\alpha$-preserving reductions, then 
 the composition  $(i,s):=(\mathop{\tt inst}_\mathcal{P}\circ \mathop{\tt inst}'_\mathcal{P},\mathop{\tt sol}'_\mathcal{P}\circ \mathop{\tt sol}_\mathcal{P})$ is also an  $\alpha$-preserving reduction.
 If both  $(\mathop{\tt inst}_\mathcal{P},\mathop{\tt sol}_\mathcal{P})$ and $(\mathop{\tt inst}'_\mathcal{P},\mathop{\tt sol}'_\mathcal{P})$ are strict (polynomial-time computable, resp.), then the composition 
 $(i,s)$ is strict  (polynomial-time computable, resp.).
\end{lemma}

\begin{proof}
Consider a situation described as in the lemma, where the pair of numbers $(a,b)$ shows that $(\mathop{\tt inst}_\mathcal{P},\mathop{\tt sol}_\mathcal{P})$ is $\alpha$-preserving
and the pair of numbers  $(a',b')$ shows that $(\mathop{\tt inst}'_\mathcal{P},\mathop{\tt sol}'_\mathcal{P})$  is $\alpha$-preserving.
Clearly, if $x\in I_\mathcal{P}$, then $\mathop{\tt inst}_\mathcal{P}(x)\in I_\mathcal{P}$
and hence $(\mathop{\tt inst}_\mathcal{P}\circ \mathop{\tt inst}'_\mathcal{P})(x)=\mathop{\tt inst}'_\mathcal{P}(\mathop{\tt inst}_\mathcal{P}(x))\in I_\mathcal{P}$, as well.
A similar observation applies to the solutions (in the reversed order).
We now prove that the composition $\mathop{\tt inst}_\mathcal{P}\circ \mathop{\tt inst}'_\mathcal{P}$
is $\alpha$-preserving, testified by the pair of numbers $(a+a',b+b')$.
\begin{eqnarray*}
m_\mathcal{P}^*((\mathop{\tt inst}_\mathcal{P}\circ \mathop{\tt inst}'_\mathcal{P})(x)) + (a+a')
&=&
(m_\mathcal{P}^*((\mathop{\tt inst}'_\mathcal{P}(\mathop{\tt inst}_\mathcal{P}(x))) + a') +a\\
&\ge& m_\mathcal{P}^*(\mathop{\tt inst}_\mathcal{P}(x))+a\\
&\ge& m_\mathcal{P}^*(x)
\end{eqnarray*}
The computation for the bounds on the solution is similar and hence omitted.
The claim on composability of the strictness and the polynomial-time computability are easy to see. 
\qed
\end{proof}

By a trivial induction argument, the previous lemma generalizes to any finite number of reductions that we like to compose.

\paragraph{Conditional reductions} In the realm of Kernelization, reductions are often described in some conditional form:
$$\textbf{if\ }\textrm{condition\ }\textbf{then\ do\ }\textrm{action}$$
Our previous considerations apply also for this type of conditioned reductions, apart from the fact that
an instance may not change, assuming that the reduction was not applicable, which means that the condition was not true for that instance.\SV{ Further discussions can be found  in the long version of the paper.}
\LV{

First, we have to make clear what the notions of ``strictness'' and ``polynomial time computations''
refer to in the context of reduction rules with conditions.
``Strictness'' now means that the input will be shortened if the condition is met, and ``polynomial time'' means two things:
a) the condition can be checked in polynomial time and b)  the possibly triggered action can be performed in polynomial time.
Moreover, often there is a finite collection of conditioned reductions.
These can be combined in quite a natural way into a single conditioned reduction.
This is formally described in the following lemma.

\begin{lemma}
 Assume that, for each $1\leq i\leq n$,
 $$\textbf{if\ }\textrm{condition}_i\textbf{\ then\ do\ }\textrm{action}_i$$
 is a conditioned $\alpha$-preserving reduction.
 Then, these can be combined into a single conditioned $\alpha$-preserving reduction $\langle \textrm{combi-condition},\textrm{combi-action}\rangle$
 as follows:
 $$\textbf{if\ }\exists i(\textrm{condition}_i)\textbf{\ then\ do\ }\textrm{perform some applicable\ action}_i$$
 If  all original conditioned reductions are strict  (polynomial-time computable, resp.), then the combined reduction is strict  (polynomial-time computable, resp.).
\end{lemma}

Now, we can present a general recipe how to obtain a polynomial-time factor-$\alpha$ approximation based on $\alpha$-preserving reductions.
The previous lemma shows that the
use of a single reduction in the formulation of the next theorem does not lose any generality.
\begin{theorem}
 Assume that $\mathcal{P}$ is some maximation problem. Suppose that
$$\textbf{if\ }\textrm{condition}(x)\ \textbf{then\ do\ }\textrm{action}(x)$$
is some conditioned $\alpha$-preserving, strict, polynomial-time computable reduction.
Further assume that there is some polynomial-time computable factor-$\alpha$ approximation algorithm $A$
for $\mathcal{P}$, restricted to instances from $\{x\in I_\mathcal{P}\mid \neg \textrm{condition}(x)\}$. Then,
there is a polynomial-time  computable factor-$\alpha$ approximation algorithm  for all instances.
\end{theorem}

\begin{proof}
 The desired algorithm should work as follows. Given an instance $x$:
 \begin{enumerate}
  \item As long as possible, some $\alpha$-preserving reductions are performed. This yields the sequence of instances
  $x=x_0$, $x_1$, \dots, $x_n$.
 
  \item Then, $A$ is applied to the reduced instance $x':=x_n$.
  \item As $y_n:=y':=A(x')$ is an $\alpha$-approximate solution for $x'=x_n$, we can successively construct $\alpha$-approximate solutions $y_{n-1}$ for $x_{n-1}$, \dots,
  $y_1$ for $x_1$ and finally
  $y:=y_0$ for $x=x_0$.
  \item Return $y$ as approximate solution for $x$.
 \end{enumerate} As compositions of $\alpha$-preserving reductions yield $\alpha$-preserving reductions (maintaining some desirable properties), as shown in Lemma~\ref{lem-compose}, any $\alpha$-approximate solution to $x_n$ can be turned into an $\alpha$-approximate solution for $x$. 
Hence, all claimed properties directly follow by our previous considerations, apart from the polynomial-time claim.
Here, observe as the reductions are strict, $n\leq |x|$, so that the \textbf{while}-loop terminates after a polynomial number of steps.\qed
\end{proof}
}

The general strategy that we follow can be sketched as follows:
\begin{enumerate}
\item Apply (strict, poly-time computable) $\alpha$-preserving reduction rules as long as possible.
\item Possibly modify the resulting graph so that it meets some requirements from known combinatorial results on the graph parameter of interest.
\item Compute some solution for the modified graph that 
satisfies the mentioned combinatorial bounds.
\item Construct from this solution a good approximate solution for the original instance.
\end{enumerate}

In order to illustrate the use of this strategy, let us elaborate on \textsc{Nonblocker}, matching a  result from \cite{NguSHSMZ2008}. Actually, conceptually this algorithm is even simpler than the one we present for \textsc{Harmless Set} in particular in the following sections. This goes along the lines of the kernelization result by Dehne \emph{et al.}~\cite{Dehetal2006}, but kernelization needs no
 constructive proof of the combinatorial backbone result;  the non-constructive proof of\SV{~\cite{McCShe89}}\LV{ \cite{Bla73a,McCShe89}} is hence sufficient.

\begin{enumerate}
\item Delete all isolates. (If the resulting graph is of minimum degree at least two, we are ready to directly apply the algorithm of Nguyen \emph{et al.}~\cite{NguSHSMZ2008}.) This rule is $\alpha$-preserving for any $\alpha\geq 1$ (with $a=b=1$).
\item Merge all leaf neighbors into a single vertex. Again, this rule is $\alpha$-preserving for any $\alpha\geq 1$ (with $a=b=0$ for a single merge and hence also for a finite sequence of merges). 
\item Delete all leaves but one, which is $x$. This yields the graph $G$ of order $n_G$.
\item Create a copy $G'$ of the graph $G$; call the vertices in the new graph by priming the names of vertices of $G$. Let $H$ be the graph union of $G$ and $G'$ plus the edge $xx'$. $H$ is of minimum degree at least two by construction. 
\item Take the algorithm of Nguyen \emph{et al.}~\cite{NguSHSMZ2008} to obtain a dominating set $D_H$ of $H$ satisfying $|D_H|\leq \frac{2}{5}n_H$.
Should the solution $D_H$ contain $x$ or $x'$, it is not hard to modify it to contain the leaf neighbors $y$ or $y'$, instead.
\item Hence, $D_G=V_G\cap D_H$ is a dominating set for $G$ with $|D_G|\leq \frac{2}{5}n_G$.
 Trivially, $N_G=V_G\setminus D_G$ is a nonblocker solution for $G$ that is $\frac{5}{3}$-approximate.
 \item As the merging and deletion reductions are $\alpha$-preserving for each $\alpha\geq 1$, we can safely undo them and hence obtain a $\frac{5}{3}$-approximate solution for the original graph instance.
\end{enumerate}

\LV{Better approximation algorithms for \textsc{Nonblocker} have been obtained by  Chen \emph{et al.}~\cite{CheENRRS2013} (with a factor of 1.41) and by 
Athanassopoulos \emph{et al.}~\cite{AthCKK2009} (with a factor of 1.244).}

%\noindent
%For a different exposition of our methodology for \textsc{Nonblocker}, see %the long version of this paper.

\LV{Nguyen \emph{et al.} used a slightly different way to obtain their approximation algorithm. 
Let us
reformulate and sketch the result from \cite{NguSHSMZ2008} within our framework.
As a reduction rule, they only remove isolates; these would be put into the nonblocker set anyways. 
The combinatorial result aimed at is the one exhibited by Blank\LV{~\cite{Bla73a}} and (independently) by  McCuaig and Shepherd~\cite{McCShe89} that shows that any graph (with seven exceptional graphs) of order $n$ with minimum degree of at least two has a dominating set with at most $\frac{2}{5}n$ vertices. This result is used by
first modifying the graph by deleting all leaves and then interconnecting the leaf neighbors so that the minimum degree two requirement is met.
It is then shown that it is possible to construct a nonblocker set for $G$, given a dominating set $D_H$ satisfying $|D_H|\leq \frac{2}{5}n_H$ for the modified graph $H$. An essential ingredient is a new proof of the mentioned result from  \cite{McCShe89} that is in fact a polynomial-time algorithm to compute $D_H$ within $H$.
This result was also used by our version of this algorithm given above.}

\section{Harmless Set}
\label{sec:Harmless}

We are now turning towards \textsc{Harmless Set} as the most elaborate example of our methodology.
 First, we are going to present the combinatorial backbone of our result. 
Let $S_2(G)$ be the set all vertices of degree two within $G$.
%Lam and Wei have shown in~\cite{LamWei2007}:

\begin{theorem}\label{LamWei} (Lam and Wei \cite{LamWei2007})
Let $G$ be a graph of order $n_G$ and of minimum degree at least two such that
$G[S_2(G)]$ decomposes into $K_1$- and $K_2$-components.
Then, $\gamma_t(G)\leq n_G/2$.
\end{theorem}

%We will discuss this result more thoroughly later in this section. However, let us mention already here that, as we are aiming at developing algorithms, t
The proof of this theorem is non-constructive, as it uses tools from extremal combinatorics.\SV{ In the long version of this paper, w}\LV{ W}e show \LV{now }how to obtain a polynomial-time algorithm that actually computes a total dominating set (TDS) $D$ with $|D|\leq n_G/2$ under the assumptions of Theorem~\ref{LamWei}.
\LV{
Let $G=(V,E)$ be a  graph with (*) minimum degree at least two and no three consecutive vertices of degree two.
This last condition is obviously equivalent to requiring that $G[S_2(G)]$ decomposes into connected components of the form $K_1$ or $K_2$. 
As connected components can be computed
consecutively, we can assume that $G$ is connected.

First, we greedily remove edges, as long as
the graph still satisfies (*).
A TDS computed for the resulting graph is also a TDS for the original graph.
For simplicity, we can hence further assume
that no edges from $G$ can be removed without
violating (*). This is a technical condition
needed for applying some of the Lemmas from \cite{LamWei2007}.

We now differentiate two main cases:
\begin{itemize}
\item If $S_2(G)$ is an independent set in $G$, i.e., $G[S_2(G)]$ has no edges, then we have to differentiate further
cases when the shortest cycle in $G$ is of length 3, 4, 5, 6, or larger.
In each of the cases, Lam and Wei show how to construct a graph $G'$ smaller than $G$ that also satisfies (*).
\item Otherwise, $G[S_2(G)]$ contains a $K_2$-component. Starting out from such a path of length one, the proof of \cite[Lemma 6]{LamWei2007} shows how to construct a set $Q$ of vertices such that the graph $G'=G[V\setminus Q]$ also satisfies (*) and, moreover, $\gamma_t(G)\leq \gamma_t(G')+
\frac{|Q|}{2}$ is satisfied.
\end{itemize}

As some optimum TDS can be surely easily computed for small graphs, the sketched procedure allows to recursively compute some TDS solution for $G$. Notice in particular that the proofs of Lam and Wei show how to construct a solution for the calling instance from the one obtained for the called instance (in the recursion).
Also, it is shown (as explicitly indicated in the second case above) that the claimed bound on the solution size easily follows by induction.

Hence, we can state the following constructive version of the combinatorial result of Lam and Wei:

\begin{theorem}\label{LamWei-constructive}
For a given graph $G=(V,E)$ of order $n_G$ that satisfies (*),
one can compute a TDS $D\subseteq V$ with
$|D|\leq\frac{n_G}{2}$ in polynomial time.
\end{theorem}

Observe that a quick analysis of the sketched algorithm indicates a bound of $O(n_G^8)$ for the running time, as one has to actually verify that there are no short
cycles in $G$ to match the case analysis.
Supposedly, a complete re-analysis of the combinatorial argument could reveal better algorithms, but for the proof of concept of our methodology, this analysis is sufficient here.
}
%We will explain at the end of this section how to turn %the result into an efficient algorithm.
%
 
 Our approximation algorithm 
%for Unanimous 
%Harmless Set (UHS) 
for \textsc{Harmless Set}
 is %(finally) 
 based on 
obtaining a (small enough) TDS in a graph $H$ obtained from
the input $G$ after a number of modifications (mainly vertex deletions).
In the reduction from $G$ to $H$, we distinguish between the number of deleted 
vertices $d$ (to get from $G$ to $H$) and the number of vertices $a$ added to 
convert the TDS $D_H$ to %the TDS 
$D_G$.

\begin{theorem}
\label{TDS2UHS}
Let $G$ be a graph of order $n_G$ and let $H$ be a graph of order $n_H$ obtained 
from $G$ by deleting $d$ vertices and possibly adding some edges. Let $D_G$ and $D_H$ be TDS solutions of $G$
and~$H$, respectively, 
such that %(i) $D_H \subseteq D_G$ and (ii)
$a = |D_G|-|D_H| \leq d$. 
If $|D_H| \leq c\cdot n_H$ and $d\leq \gamma_t(G)$, then $V(G)\setminus D_G$ is a %unanimous 
harmless
set of $G$ whose size  $n_G-|D_G|$ is within a factor of $(1-c)^{-1}$ from optimum.
\end{theorem}

\begin{proof} 
As $n_H=n_G-d$,
$|D_G| = |D_H|+a \leq c(n_G-d)+a = cn_G+(a-cd) 
\leq cn_G+ d-cd = cn_G + (1-c)d \leq cn_G + (1-c)\gamma_t(G)$.
Hence,  $n_G-|D_G| \geq n_G - cn_G - (1-c) \gamma_t(G) = (1-c)(n_G - \gamma_t(G))$.
This immediately yields an approximation factor of $(1-c)^{-1}$.\qed
\end{proof}

In the following\LV{ section}, we will present reduction rules that produce a graph
$G$ with the property (**) that each vertex of degree bigger than one has at most one leaf neighbor. The surgery that produces a graph $H$ from $G$ as indicated in Theorem~\ref{TDS2UHS} includes removing all $d$ leaves and adding edges to ensure that $H$ has minimum degree of two and satisfies that each component of $H[S_2(H)]$ has diameter at most one.
Notice that all leaf neighbors  in $G$ belong to some optimum TDS of $G$ without loss of generality. Due to (**), $\gamma_t(G)\geq d$ as required.
Moreover, given some TDS solution $D_H$ for $H$, we can produce a valid TDS solution $D_G$ for $G$ by adding all $d$ leaf neighbors to $D_H$.
Notice that Theorem~\ref{TDS2UHS}  leads to a factor-2 approximation algorithm for \textsc{Harmless Set} based on\SV{  a polynomial-time, constructive version of} Theorem~\SV{\ref{LamWei}}\LV{\ref{LamWei-constructive}}.

%\marginpar{To Do!!}

%\begin{algorithm}[H]
% \KwData{A graph~$G=(V,E)$ with minimum %degree at least two and no three-%consecutive vertices of degree two.}
% \KwResult{A factor-2 total dominating %set~$D$.}
% \Begin{
%  \eIf{$|V| = O(1)$}{
%   Compute a MTDS~$D$ by brute-force\;
%   \KwRet{$D$}\;
%   }{
%    $(G', T) \gets CasesAnalysis(G)$\tcc*%[r]{$V(G') \subset V$ and $T \subseteq V$.}
%	$T' \gets LamWeiApprox(G')$\;
%	\KwRet{$T' \cup T$}\;
%  }
%  
% }
% \caption{Procedure $LamWeiApprox$}
% \label{algo:2tds}
%\end{algorithm}

%given a graph that does not contain two consecutive vertices of degree two,
%returns a total domination set that contains no more than %half of the vertices of the graph.
%Our goal is to apply the previous theorems, so we 
%\begin{enumerate}
% \item first apply factor-2 preserving reduction rules to %obtain from some graph $\Gamma$ a graph $G$ with some special properties;
% \item  then, we delete some vertices and add some edges to produce some graph $H$ with no two consecutive vertices of degree two;
% \item applying further factor-2 preserving reduction rules and then 
% the algorithm of Lam and Wei yields some UHS solution $U_H$ containing at least half of the vertices of $H$; 
% \item Theorem~\ref{TDS2UHS} allows to construct a UHS solution $U_G$ for $G$ that is factor-2 approximate;
% \item undoing the factor-2 preserving reductions finally allows to produce a factor-2 approximate soluation $U_{\Gamma}$ for the original instance $\Gamma$.
%\end{enumerate}

\LV{
In the following section, we are going to describe the reduction rules necessary to produce a graph to which we could apply the mentioned combinatorial results.

\section{Reduction Rules for \textsc{Harmless Set}}
\label{sec:Harmless-Reductions}}
%
%In the rest of the paper, we denote by $N(v)$ the set of neighbors of 
%vertex $v$ and by $N[v] = N(v)\cup \{v\}$ the closed neighborhood of $v$.
%
%
%We shall use the expression {\em reduction rule} to %describe a 
%pair $\left<\textit{condition},\textit{action}\right>$ of %a reduction
%procedure that applies $\textit{action}$ when %$\textit{condition}$ holds.\marginpar{possibly reword}
%Such rules have been the major tool within parameterized algorithmics. We now 
%show how to use such ideas in the area of approximability. 
Now, we list $\alpha$-preserving reductions for \textsc{Harmless Set}. 
% \LV{All missing  correctness proofs can be found in the long version of this paper. }
We start with two very simple rules.

\paragraph{\bf Isolate Reduction\SV{.}}
If there is some isolated vertex, produce the instance $(\{x\},\emptyset)$
that has trivially no solution.
If there is some isolated edge $xy$, produce that instance $G[V\setminus \{x,y\}]$ from $G=(V,E)$.
%Delete isolates. (They would go into the TDS.)

For the correctness of this rule, observe that a graph with isolated vertices has no total dominating set at all.

\paragraph{\bf Leaf Reduction\SV{.}} If there are two leaf vertices $u,v$ with common neighbor $w$, then delete $u$. (It would go into the harmless set.)

%As shown in the appendix, this rule easily generalizes to deal with all twins, but this is not exactly necessary for our algorithmic approximation result.

\begin{observation}\label{obs-harmless-isolates-leaves}
The 
Isolate Reduction (for edges) and the 
Leaf Reduction are $\alpha$-preserving for any $\alpha\geq 1$.
\end{observation}

\LV{\begin{proof} 
%(of Observation~\ref{obs-harmless-isolates-leaves})
%The isolate rule is $\alpha$-preserving by setting $a=b=0$ in the definition.
%In other words, isolates must belong to any TDS solution.
The 
Isolate Reduction  is $\alpha$-preserving by setting $a=b=0$ in the definition.
In other words, endpoints of isolated edges must belong to any TDS solution.
The Leaf Reduction is $\alpha$-preserving by setting $a=b=1$ in the definition.
In other words, w.l.o.g., leaves do not belong to some TDS solution, except when there is a $K_2$-component in the graph. %In that case, the Isolate Reduction would first put one of the two leaves into the harmless set solution and turn the other one into an isolate, to which
%the isolate rule can apply thereafter, putting it into the TDS.
\LV{\qed}\end{proof}}

Hence from now on, no vertex can have two leaf neighbors.

\LV{Actually, we could generalize the 
Leaf Reduction towards the following rule:

\paragraph{\bf Twin Reduction}
Recall that vertices $u$ and $v$ are said to be {\em true twins} if $N[u] = N[v]$ and
{\em false twin} if $N(u) = N(v)$. 

\begin{itemize}
\item If there are two vertices $u$ and $v$ such that $N[u]=N[v]$, i.e., they form true twins, then delete $v$
(it would go into the harmless set).
\item If there are two vertices $u$ and $v$ such that $N(u)=N(v)$, i.e., they form false twins, then delete $v$
(it would go into the harmless set).
\end{itemize}

As we are not using this rule in some crucial manner in what follows, we 
present the following result without proof. 

\begin{theorem}
The Twin Reduction is $\alpha$-preserving for any $\alpha\geq 1$.
\end{theorem}
}

%We shall use the term {\em chain} to denote an induced path whose interior 
%vertices are of degree two in $G$. 
A 
%vertex of degree one in $G$ is 
%{pendant} and a 
chain with one leaf endpoint is a {\em pendant chain}.
A {\em floating chain} is a chain with two leaves.
A {\em support vertex} is a non-pendant endpoint of a pendant chain. %, examples being leaf neighbors. 
 Support vertices may have more than one pendant chain. We shall reduce
the length of pendant chains to at most two, based on the following reduction
rules. The first one actually generalizes the Isolate Reduction.

\paragraph{\bf Floating Chain Reduction\SV{.}} Delete all floating chains. 

% Clearly, a maximum harmless set (and a minimum TDS) can be computed in linear time 
% on such trivial connected components.\LV{ (Moreover, none of the other reduction 
% rules introduce a floating chain.)} Hence, we can verify the definition with suitably chosen values for $a=b$, which proves:

\begin{observation}
\label{floating-path}
The Floating Chain Reduction\LV{ rule} is $\alpha$-preserving for any $\alpha\geq 1$.
\end{observation}

\LV{\begin{proof} 
%In the formulation of Theorem~\ref{thm-alphaapprox-general},
$G'$ is obtained from
$G$ by deleting a floating chain. For the chain, the numbers $a=b$ can be
computed (optimally) in polynomial time; they correspond to the size of
optimum solutions for the floating chain component.\qed
\end{proof}}

\paragraph{\bf Long Chain Reduction\SV{.}}
Assume that $G$ is a graph that contains a path $x-u-v-w-y$, where $u,v,w$ are three consecutive vertices of degree two, where $|N(y)|\geq 2$. Then, construct the graph $G'$ by
\LV{
\begin{itemize}
\item deleting $x,u,v,w$ and
 
\item connecting $y$ to all vertices in $N(x)\setminus \{u\}$ (without creating double edges).
\end{itemize}

This corresponds to }merging $x$ and $y$ and deleting $u,v,w$.
\LV{This Long Chain Reduction resembles the folding rule known for \textsc{Vertex Cover} (in Parameterized Complexity, see~\cite{DowFel2013}).}

\begin{theorem}
\label{thm-chain}
The Long Chain Reduction is $\alpha$-preserving for any $\alpha\geq 1$.
\end{theorem}

\begin{proof}
% We shall verify our definition; to this end, we show 
% that $a=b=2$ works out in our case.
Let $G$ be the original graph and $G'$ the graph obtained from $G$ by deleting the path $u,v,w$ and merging $x$ and $y$ as described by the rule. We show that 
$a=b=2$ works out in our case by considering several cases.
%By $N(q)$ for any vertex $q\in V$, we always refer to the neighborhood in $G$.

\smallskip\noindent
(a) Let $C$ be a maximum harmless set (HS) for $G$. \SV{The special case when $N(x)=\{u\}$ is easy to verify. }\LV{Let us first briefly discuss what 
happens if $N(x)=\{u\}$.
Then, it is not hard to see that an optimum solution $C$ would contain $x$ 
and $w$, but not $u$ and $v$.
Merging $x$ and $y$ and deleting $u,v,w$ is now equivalent to deleting the 
whole pending path $x-u-v-w$. As $w\in C$, it does not dominate $y$, so that 
%%% I changed v to w:
$C'=C\setminus \{x,w\}$ is a valid harmless set for $G'$.
 
}In the following discussion, we can hence assume that $x$ has at least two 
neighbors. We now consider cases whether or not $x\in C$ or $y\in C$.

\LV{\begin{itemize}

\item}\SV{\smallskip\noindent (a1)} Assume that $x\in C$ and $y\in C$. Hence, $u,w$ are not dominated neither 
by $x$ nor by $y$. As $C$ is maximum, we can assume
$|C\cap \{u,v,w\}|=1$, as  
$\min\{|N(x)|,|N(y)|\}\geq 2$; hence, if
all of $u,v$ and $w$ are in $V\setminus C$, then we can replace $w$ by 
another neighbor of $y$ and obtain another optimum solution. 
Then, $C'=C\setminus\{x,u,v,w\}$ is a HS of $G'$, with $|C'|=|C|-2$.

\LV{\item}\SV{\noindent (a2)} Assume that $x\notin C$ and $y\notin C$. First, let us discuss the 
possibility that $u\notin C$ and $w\notin C$. As $C$ is maximum, the purpose 
of this is to dominate (i) $v$ and (ii) $x$ and $y$. To accomplish (i), either 
$u\notin C$ or $w\notin C$ would suffice.
However, as $C$ is maximum, condition (ii) means that $N(x)\setminus C=\{u\}$ 
and that $N(y)\setminus C=\{w\}$. By our assumptions, 
$\min\{|N(x)|,|N(y)|\}\geq 2$. Hence, there is  a vertex $z\in N(y)$, $z\neq w$. 
Now, $\tilde C=(C\setminus\{z\})\cup\{w\}$ is also a maximum HS satisfying $\{v,w\}\subseteq \tilde C$. From now on, we assume that $|C\cap \{u,v,w\}|=2$ and 
that $|((N(x)\cup N(y))\setminus (\{u,w\}\cup C)|\geq 1$\LV{ (in other words,
at least one of $x$ and $y$ has a neighbor in $V\setminus C$ other than
$u$ and $w$, respectively)}. 
Hence, $C'=C\setminus\{u,v,w\}$ is a HS of $G'$ with $|C'|=|C|-2$.

\LV{\item}\SV{\noindent (a3)} Assume now that $x\in C$ and $y\notin C$. (Clearly, the case that 
$x\notin C$ and $y\in C$ is symmetric.)
As $u$ is not dominated by $x$, either (i) $\{u,v\}\subseteq V\setminus C$ or 
(ii) $\{v,w\}\subseteq V\setminus C$.
In case (i), $x$ is dominated by $u$, but $y$ must (still) be dominated by 
some vertex from $N(y)\setminus \{w\}$.
In case (ii), symmmetrically $y$ is dominated by $w$, but $x$ must be 
dominated by some vertex from $N(x)\setminus \{u\}$.
In both cases, $\tilde C=(C\setminus \{x\})\cup\{v\}$ is another maximum 
harmless set of $G$.
This leads us back to the previous item (i.e., $|C'|=|C|-2$.)
\LV{\end{itemize}}

Summarizing, we have shown that from $C$ we can construct a harmless set 
$C'$ for $G'$ with  $|C'|=|C|-2$.

\smallskip\noindent
(b) Conversely, assume $C'$ is some harmless set for 
$G'$. We distinguish two cases:

\LV{\begin{itemize}

\item}\SV{\smallskip\noindent (b1)} Assume that $y\in C'$. Then, $y$ is dominated by some $z$ in its 
neighborhood (in $G'$). We consider two cases according to the situation 
in $G$.
(i) If $z\in N(x)$, then $C=C'\cup\{x,u\}$ is a HS in $G$. 
(ii) If $z\in N(y)$, then $C=C'\cup\{x,w\}$ is a HS in $G$.
In both cases, $|C|=|C'|+2$.

\LV{\item}\SV{\noindent (b2)} If $y\notin C$, then again $y$ is dominated by some $z$ in its 
neighborhood (in $G'$). We perform the same case distinction as in the 
previous case:
(i) If $z\in N(x)$, then $C=C'\cup\{u,v\}$ is a HS  in $G$. 
(ii) If $z\in N(y)$, then $C=C'\cup\{v,w\}$ is a HS in $G$.  
In both cases, $|C|=|C'|+2$.

\LV{\end{itemize}}

\noindent
(c) The reasoning from (b) shows that, if $C$ is an optimum solution for $G$, 
then $C'$ as obtained in part (a) of this proof is an optimum solution for $G'$.\LV{
Namely, assume that there would be a harmless set $C^*$ for $G'$ with 
$|C^*|>|C'|$. Then, according to (b), we can construct a harmless set of $G$ 
with $|C^*|+2>|C|$ many vertices, contradicting the  maximality of $C$. }\qed
\end{proof}

Similarly, one sees the correctness of the 
%\SV{\textbf{Cycle Chain Reduction.}}\LV{
following rule.
%}\LV{
\paragraph{\bf Cycle Chain Reduction\SV{.}}
If $G$ is a graph that contains a cycle $x-u-v-w-x$, where $u,v,w$ are three consecutive vertices of degree two, then construct\LV{ the graph} $G'$ by
 deleting $u$.\LV{
 \begin{observation}
 The cycle chain reduction  is $\alpha$-preserving for any $\alpha\geq 1$.
 \end{observation}
 
 \begin{proof}
 An optimum harmless set for $G$ will put exactly two out of the three vertices $u,v,w$ into the harmless set.
 W.l.o.g., let these be $u$ and $v$. Conversely, $w$ and $x$ would go into the total dominating set. Also, in the reduced graph, $v$ will be in the harmless set, while $x$ and $w$ will be in the total dominating set. This shows the claim with constants $a=b=1$.\qed 
 \end{proof}}
 
 \SV{\bigskip}
 Finally, we deal with support vertices with multiple pendant chains. Assuming
the Long Chain Reduction has been applied, any pendant chain is of length
two or less. Accordingly, a support vertex
%\marginpar{attachment? support?} 
where two of more pendant chains
meet does belong to some optimum solution. The following rule 
makes this idea more precise.
%therefore $\alpha$-preserving for $\alpha \geq 1$.

\paragraph{\bf Pendant Chain Reduction\SV{.}}
Assume that $G=(V,E)$ is a graph that contains two pendant chains with
common endpoint $v$ of which at least 
one path is of length two. Then, construct the graph 
$G'=(V',E')$ by deleting one of the two pendant chains, keeping one which is of length two.

\begin{theorem}
\label{thm-pendant}
The Pendant Chain Reduction %rule
is $\alpha$-preserving for any $\alpha\geq 1$.
\end{theorem}

\LV{\begin{proof} 
Let $v-x-y$ and (a) $v-z-t$ (or (b) just $v-t$) be two pendant paths of $G$.
Then $y$ belongs to some maximum harmless set $C$ of $G$ while $x$ belongs
to $V(G)\setminus C$. Similarly, $z$ (if existent) belongs to some maximum harmless set $C$ of $G$ while $t$ belongs
to $V(G)\setminus C$.
It follows that the definition of $\alpha$-preserving reduction can be applied with
$a=b=1$. Notice that, because we keep $v-x-y$,
neither $v$ nor $x$ will belong to any harmless set solution for the reduced graph. Hence, adding $t$ to the harmless set solution of the reduced graph is always possible, resulting in a valid harmless set for the original graph.
\qed
\end{proof}
}

We are now \SV{ready}\LV{in the position} to apply  Theorem \ref{TDS2UHS}.\LV{

\begin{observation}}
Assume the graph $G=(V,E)$ is reduced according to the reduction rules 
described so far.
Hence, $G$ satisfies\LV{ the following properties}:
\LV{
\begin{itemize}
\item}\SV{(a)}  $G$ contains no chain of three vertices of degree two. \LV{
\item}\SV{(b)} By the Leaf Reduction rule, any vertex has at most 
one leaf neighbor.
%vertex of degree 
%one is attached.
\LV{\end{itemize}
\end{observation}
}
%The unique neighbor of a leaf is called an attachment vertex.
%We assume in the following construction that $G$ contains at least two 
%(and an even number thereof) leaf neighbors.
Let $G'$ be a graph isomorphic to $G$ so that each vertex
$v$ of $G$ corresponds to a vertex $v'$ of $G'$, under the assumed 
isomorphism $f: V(G) \longrightarrow V(G')$.
 We construct a graph $H$ obtained from the disjoint union of $G$ and 
$G'$ simply by adding edges between each leaf neighbors vertex $v$ of $G$ 
with $v' = f(v) \in V(G')$.
Then, we remove all leaves.
%
%Notice that each vertex of degree one has a ``personal'' attachment 
%vertex, i.e., the mapping from leaves to attachment vertices is injective 
%due to the first property of $G$. 

Due to the application of\LV{ the} Pendant Chain Reduction\LV{ rule} to $G$ 
(and $G'$),
the addition of edges between corresponding  leaf neighbors in $G$ and $G'$
does not introduce induced cycles with more than two 
consecutive degree-two vertices.

To the resulting graph $H$, apply Long Chain Reduction as long as possible. 
Notice that an application of this rule does never decrease degrees, adds two 
vertices to the solution and removes four vertices of the graph. 

This results in a graph $H'$ of order $n_{H'}$ with minimum degree at 
least two containing  no chain of three vertices of degree two.
%no three consecutive vertices of degree two.
Hence, we can apply the (algorithmic) version of Theorem~\ref{LamWei} that 
returns a TDS $D_{H'}$ for $H'$ with $2|D_{H'}|\geq n_{H'}$.
Undoing a certain number of  Long Chain Reductions, say,  $c$, that we applied, we obtain a TDS $D_H$ for 
$H$ with $2|D_H|=2(|D_{H'}|+2c)\geq n_{H'}+4c=n_H.$
By symmetry, we can assume that $|D_H\cap V(G)|\leq |D_H\cap V(G')|$. 
Now, we add all support vertices to $D_H\cap V(G)$ and further 
vertices to obtain $D_G$ by the following rules:

\begin{itemize}
\item 
If a support vertex already belongs to $D_H$, then it could have been 
dominated via the edge that we introduced.
As this interconnects to another support vertex, both already belonged to 
$D_H$. We arbitrarily select two neighbors (in $G$) of these support 
vertices and put them into $D_G$. Hence, the mentioned support vertices and the attached leaves are  totally dominated. 

\item If a support vertex $x$ did not already belong to $D_H$, two cases 
arise: (a) If it was dominated (in $H$) via an edge already belonging to $G$, 
then we do nothing on top of what we said. (b) If the support vertex $x$ 
was dominated (in $H$) by an edge $xy$ not belonging to $G$, then we must add 
another neighbor $z$ (in $G$) of $x$ to $D_G$. However, as (obviously) the 
vertex $y$ belonged to $D_H$ and was dominated by a neighbor (in $G$) in 
$D_H$, we add (in total) two vertices $x,z$ for the two support vertices 
$x,y$. Seen from the other side, this covers the case of a support vertex 
that already belonged to $D_H$ but was not dominated via the edge that we 
introduced.
\end{itemize}

Altogether, we see that we delete all leaves and  introduce at 
most that many vertices into $D_G$ (in comparison to $D_H\cap V(G)$).
By Theorem~\ref{TDS2UHS} and since all reduction rules take polynomial time, we obtain:

\begin{theorem}
\textsc{Harmless Set}  is factor-$2$  polynomial-time approximable.\LV{\qed}
\end{theorem}

%We can compute a total dominating set $D_{G\cup G}$ of %$G\cup G$. The 
%complement of the smaller of the two component parts, %call it $D_G$,
%can be used to approximate \textsc{Harmless Set} with %the same 
%approximation factor as stated before.

\section{The differential of a graph}
\label{sec:Differential}

Let us start with an alternative presentation of this notion. Let $G=(V,E)$ be a graph. For $D_0\subseteq V$, let 
\LV{$}$ \partial (D_0):=
\left|\left(\bigcup_{x\in D_0}N(x)\right)\setminus D_0\right|-|D_0|.$\LV{$}
$\partial(D_0)$ is called the \emph{differential} of the set $D_0$, and our aim is to find a vertex set that maximizes this quantity.
This maximum quantity is known as the differential of $G$, written $\partial(G)$.
The following combinatorial results are known:

\begin{theorem} \label{Thm-differential-combinatorics} \cite{BerFer2012}
Let $G$ be a connected graph of order $n$.
\LV{\begin{itemize}
\item}\SV{(a)} If $n\geq 3$, then $\partial(G)\geq n/5$.
\LV{\item}\SV{(b)} If  $G$  has minimum degree at least two,
then
$\partial(G)\geq \frac{3n}{11}$,
apart from five exceptional  graphs, none of them having more than seven vertices.
\LV{\end{itemize}}
\end{theorem}

It is not hard to turn the first combinatorial result into a kernelization result, yielding a 
kernel %size 
bound of $5k$, where $k$ is the natural parameterization of the \textsc{Differential}.\LV{ Along the lines of \cite{BazCho2012}, we can obtain a factor-5 approximation by first computing a spanning tree $T=(V,E_T)$ for $G$ and then computing an optimum differential set $D_T$ in $T$ by dynamic programming,
and then observing that $D_T$ is a factor-5 approximation for $G$.}
In \cite{BerFer2013sub}, this result was improved to a kernel whose order
%size (number of vertices) 
is bounded by $4k$. \LV{Along those lines,}\SV{This way,} we can also get a factor-4 approximation.
However,\LV{ the second item of} Theorem~\ref{Thm-differential-combinatorics} suggests a possible improvement to a factor of $\frac{11}{3}$ \LV{if we employ}\SV{by} our framework.\LV{ This is what we are going to endeavor in this section.}

First, we have to show\SV{ (see the long version of this paper for more details)} that the reduction rules presented in~\cite{BerFer2013sub} as kernelization rules can be also interpreted as $\alpha$-preserving rules.
We use some non-standard terminology\LV{ for stating the rules}. 
%\marginpar{Possibly, we have to %	reword.}
%A \emph{leaf} is a vertex of degree one and a
A \emph{hair} is a
sequence of two vertices $uv$, where $u$ is a leaf and $v$ has
degree two. Then, $u$ is also called a \emph{hair leaf}. 

\LV{
We use the
following simple notation for a hair $uv$ for reasons of clarity:
$u-v-\cdots$.

\begin{enumerate}
 \item {\bf Leaf Reduction.} If there are two leaves connected to the same vertex, then connect these leaves.
%The parameter remains unchanged.
 \item {\bf Hair Reduction.} If there are two hairs connected to the same vertex, then remove the two hair leaves.
%The parameter remains unchanged.
\item {\bf Leaf-Hear Reduction.} If there is a leaf and a hair connected to the same vertex, then remove the hair leaf.
%The parameter remains unchanged.
\item {\bf Long Hair Reduction.} If there is a hair $u-v-\cdots$ connected to a vertex $w$ of degree two, then
remove $u,v,w$. % and decrease the parameter by one.
\item {\bf Neighbor Hair Reduction.} If there is a hair $u-v-\cdots$ connected to a vertex $w$ and another hair
$u'-v'-\cdots$ connected to a neighbor $w'$ of $w$, then remove the
edge $ww'$. %The parameter remains unchanged.
\end{enumerate}

\begin{figure}[tbh]
\begin{center}
 \includegraphics[width=.9\textwidth]{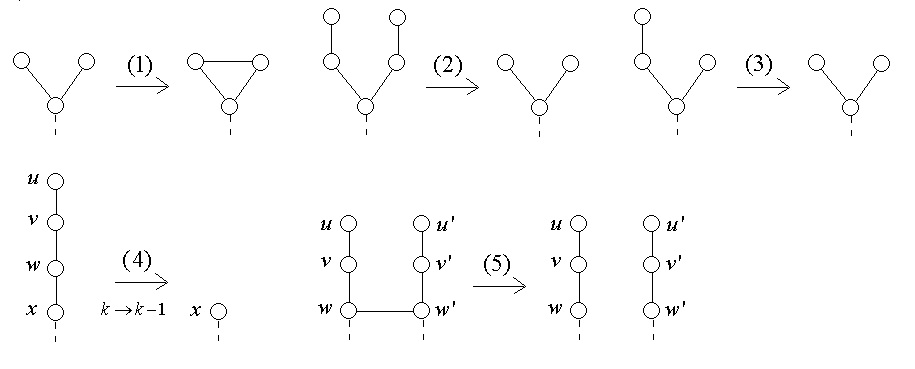}
 \end{center}
\caption{Reduction rules for \textsc{Differential}.\label{fig-Kernel4_14}}
\end{figure}

In the reasoning given for the rules in~\cite{BerFer2013sub}, only for the Long Hair Reduction, the natural parameter changes (decreases by one). The argument shows (for the other cases) that even if a set of vertices is produced for the reduced graph that is not a
valid solution for the original graph, still
another solution can be constructed that is not
worse (smaller) than the one that was obtained,
so that approximation factors are clearly preserved. The Long Hair Reduction can be seen to be $\alpha$-preserving when setting $a=b=1$.
Hence, we can summarize:

\begin{observation}
The previous five reductions are $\alpha$-preserving for any $\alpha\geq 1$.
\end{observation}

\begin{lemma} \cite{BerFer2013sub} \label{lem-reduced-properties}
 Let $G=(V,E)$ be a graph where none of the previous five reductions  applies.
Then, $G$ has the following properties:
\begin{enumerate}
 \item[\emph{(1)}] To each vertex, at most one leaf  or one hair is attached, but not both together.
 \item[\emph{(2)}] If we remove all leaves and all hairs from $G$, then the remaining graph
 $\tilde G=(\bar V, \bar E)$, henceforth called \emph{nucleus},
has minimum degree of at least two.
\item[\emph{(3)}] If a hair is attached to a vertex $u$ in the nucleus, then no hair is attached to any neighbor of $u$
within the nucleus.
\end{enumerate}

\end{lemma}

Notice that the properties listed in Lemma~\ref{lem-reduced-properties} ensure that when obtaining the nucleus $H$ from the reduced graph $G$ by deleting $d$ vertices, $d\leq \gamma_R(G)$ is verified.
In order  to verify that a sufficiently big solution for the nucleus can be found in polynomial time, observe the proof strategy of \cite{BerFer2012}: There, the differential of a graph is modeled by so-called big star packings.
It is possible to start with a greedily obtained big star packing and then further modify the solution, using the local (and hence easy-to-check) criteria exhibited in various lemmas of that paper, up to the point when no further improvements are possible.
The big star packing obtained in this way corresponds to a differential set $D$ with $\partial(D)\geq \frac{3}{11}n$, where $n$ is the order of the graph.

As the proof in \cite{BerFer2012} uses extremal combinatorial arguments, it is (at least at first glance) non-constructive.
Let us give some more details of the algorithm that is hidden within these  combinatorial arguments in the following.

\paragraph{The greedy selection of a big star packing.}

We will (from now on) work on a mixed graph (i.e., a graph that has both directed and undirected edges; directed edges are also called arcs) such that each vertex has at most one outgoing  arc, and no vertex with an incoming arc has an outgoing arc.
We will call a vertex incident to some directed arc \emph{marked}.
As we start with an undirected graph $G=(V,E)$, at the beginning all vertices are unmarked.
We proceed as follows:

\smallskip\noindent
As long as possible:
\begin{itemize}
\item Pick some unmarked vertex $x$ with at least two unmarked neighbors.
\item Direct all edges connecting $x$ to any unmarked neighbor towards $x$.
\end{itemize}

Now, consider the set $D$ of vertices to which some arcs point to, and let $B(D)$ denote the remaining marked vertices. Clearly, $\left(\bigcup_{x\in D}N(x)\right)\setminus D=B(D)$. Hence, $\partial(D)=|B(D)|-|D|$.
Moreover, due to the directions of the edges,
we can view each $x\in D$ as the center of a star to which at least two arcs (rays) are pointing.
So, we have defined a collection $\mathcal{S}(D)$ of stars that can be viewed as a star packing. As each star has at least two rays, we called them \emph{big stars}. Due to our greedy approach,
we hence arrive at $\mathcal{S}(D)$ as being a maximal big star packing.
Moreover, $|B(D)|$ is also the number of
directed edges (or rays) in total.
Let $C(D):=V\setminus (B(D)\cup D)$.

By definition of the partition $(D,B(D),C(D))$ of $V$ we find:

\begin{observation}\label{obs-dif0}
No edge connects vertices from $D$ with vertices from $C(D)$.
\end{observation}

As we obtain a maximal big star packing, we conclude:

\begin{observation}\label{obs-dif1}
 The induced graph $G[C (D)]$ is undirected and decomposes into $K_1$- and $K_2$-components.
\end{observation}

\paragraph{First local improvement.}
We are now going to improve the solution found so far.

\smallskip\noindent
As long as possible:
\begin{itemize}
\item Pick some vertex $x$ from $B(D)$ that has two or more neighbors in $C(D)$.
\item Let $y\in D$ be such that the edge $xy$ is directed towards $y$.
\item Replace the arc from $x$ to $y$ by an undirected edge again.
\item If there is now (only) one arc  $zy$ directed to $y$, remove $y$ from $D$ and render $zy$ an undirected edge again.
(This will increase the number of unmarked vertices.)
\item Direct all edges that connect $x$ to some umarked vertex towards $x$ and put $x$ into $D$.
 \end{itemize}

If no further improvements are possible, one might want to ensure that the (new) big star packing $\mathcal{S}(D)$ is still
maximal. If not, obvious further improvements are possible.
However, after a finite number of steps, this will end.

The new set $D$ (and the related star packing
 $\mathcal{S}(D)$ that can be read off from the directed edges) satisfies Observation~\ref{obs-dif1}
and:

\begin{observation}\label{obs-dif2}
 Every $x\in B(D)$ has at most one neighbor in $C(D)$.
\end{observation}

\paragraph{Second local improvement.}
By a procedure similar to the previous case,
we can create new stars if some $x\in B(D)$ is
part of a star with at least three rays and
neighbor of some $K_2$-component in $G[C(D)]$.
Leaving out details in this case, we can observe
for the (new) differential set $D$:

\begin{observation}\label{obs-dif3}
 If $x\in B(D)$ is neighbor of some $K_2$-component in $C(D)$, then it belongs to
 some star with at least three rays.
\end{observation}

Some simple computations (as undertaken in \cite{BerFer2012}) show that the set $D$ satisfies our desired bound, i.e., $|D|\geq \frac{3}{11}|V|$, if the packing $\mathcal{S}(D)$
only contains stars with at least four rays.
Anyways, it could well be that the (valid) differential set $D$ satisfies the bound and we can stop here.

\paragraph{Further local improvements on smaller stars $K_{1,2}$.}
If not, then we have to make further local improvements on these smaller stars, considering them in groups. The (relatively messy) details can be found in Lemmas 3.10 through 3.17 in~\cite{BerFer2012}, but this should make clear that finally we can obtain a sufficiently big differential in polynomial time.

Having obtained such differential set for the nucleus of a graph, this solution can be easily lifted to a solution of the reduced graph;
Theorem~\ref{RDS2Diff} allows us to conclude Theorem~\ref{Differential}.
}

\noindent
\LV{We are going to use the idea of computing a sufficiently big solution for the nucleus, based on the following variant of Theorem~\ref{TDS2UHS}.}
\SV{We compute a sufficiently big solution for the nucleus and then use:}

\begin{theorem}
\label{RDS2Diff}
Let $G$ be a graph of order $n_G$ and let $H$ be a graph of order $n_H$ obtained 
from $G$ by deleting $d$ vertices.
%(and possibly adding some edges). 
 Let $D_G=D_{G,1}\cup D_{G,2}$ and $D_H=D_{H,1}\cup D_{H,2}$ be Roman DS solutions of $G$
and~$H$, \LV{respectively, such that}\SV{with} %(i) $D_{H,1} \subseteq D_{G,1}$, 
$D_{H,2}=D_{G,2}$ and %(ii)
$a = |D_{G,1}|-|D_{H,1}| \leq d$. 
If $|D_{H,1}|+2|D_{H,2}| \leq c\cdot n_H$ and $d\leq \gamma_R(G)$, then $\partial(V(G)\setminus D_G)=n_G-2|D_{G,2}|-|D_{G,1}|$ is within a factor of $(1-c)^{-1}$ from optimum.
\end{theorem}
\begin{proof}
As $n_H=n_G-d$,
$|D_{G,1}|+2|D_{G,2}| = |D_{H,1}|+a +2|D_{H,2}|\leq c(n_G-d)+a = cn_G+(a-cd)
\leq cn_G+ d-cd = cn_G + (1-c)d \leq cn_G + (1-c)\gamma_R(G)$.
Hence,  $n_G-2|D_{G,2}|-|D_{G,1}| \geq n_G - cn_G - (1-c) \gamma_t(G) = (1-c)(n_G - \gamma_R(G))=(1-c)\partial(G)$.
This immediately yields an approximation factor of $(1-c)^{-1}$.\qed
\end{proof}

We can turn the (non-constructive) combinatorial reasoning of \cite{BerFer2012} into a polynomial-time algorithm\SV{ (see~\cite{AbuBCF2014})}, which allows us to conclude\LV{ with our framework}:

\begin{theorem}\label{Differential}
\textsc{Differential} is factor-$\frac{11}{3}$  polynomial-time approximable.
%\textsc{Differential} can be approximated up to a factor of $\frac{11}{3}$ in %polynomial time.
\end{theorem}

\section{Multiple Nonblocker sets}
\label{sec:Multiple}

%Combinatorial upper bounds for this parameter are collected in~\cite{CocGamShe85}.

%A \emph{$k$-nonblocker set} in a graph $G$ is the complement of a $k$-dominating set of $G$. 
We are first going to explain why neither some nice approximation algorithm nor some FPT algorithm (with the standard parameterization) yields useful results.
We shall assume $k > 1$ in this section.
% We only sketch the reduction in the following, details should be doable by the reader.
%This reduction shows:

\begin{theorem}
$k$-\textsc{Dominating Set}, $k>1$ cannot be better approximated than \textsc{Dominating Set}.
Likewise, the (standard) parameterized version is W[2]-hard.
\end{theorem}
\begin{proof}
Namely, given an instance $G$ of \textsc{Dominating Set}, we introduce (in total) $k$ copies  of each vertex, say, $v[1],\dots,v[k]$ of vertex $v$, and introduce a $K_{k,k}$ in
$\{u[1],\dots,u[k]\}\cup \{v[1],\dots,v[k]\}$ whenever there is an edge $ uv$ in $G$.
Then, the new graph has a $k$-dominating set of size $kt$ if
and only if the original graph $G$ has a dominating set of size $t$.
\end{proof}

We consider now a combinatorial upper bound on the size of some feasible solution of the minimization problem.

\begin{theorem}[\cite{CocGamShe85}]
Let $G$ be a graph of order $n_G$ and a minimum degree at least $k$. Then $\gamma_k(G)\leq \frac{k}{k+1}n_G$.
\end{theorem}

The known non-constructive proof can be turned into a polynomial-time algorithm
obtaining the following result.
% (we refer to the long version of this paper) computing a $k$-dominating set $D$ with $|D|\leq \frac{k}{k+1}n_G$. 
%\LV{
%
%\medskip
%We describe a polynomial time algorithm in what follows. 
%Let $G=(V,E)$ be a graph with minimum degree at least $k$. 
%First, we  remove edges between  vertices of degree greater %than  $k$ obtaining a graph $G'$ of minimum degree $k$. 
% Let $S = \{v\in V: d(v) > k\}$. Then $S$ is an independent %set in $G'$.
%Construct a maximal independent set $T$ that contains $S$. %Then $V\setminus T$ is a $k$-dominating set.
%If $|V \setminus T| \leq kn/(k+1)$, then $D=T$.
%Otherwise, while $|V \setminus T| > kn/(k+1)$ construct a %maximal independent set $T'$ of $G[V \setminus T]$ and set $T=T'$.%
%
%
%We show in the following that the algorithm finishes. When %$|V \setminus T| > kn/(k+1)$, denote $r = |T|$. Thus %$n=r+|V\setminus T| > r+ kn/(k+1)= r+ k(r+|V \setminus %T|)/(k+1)$, and we get $|V \setminus T| > kr$.
%Since every vertex of $V\setminus T$  has degree less than %$k$ in $G'[V\setminus T]$, any maximal independent set of %$G'[V\setminus T]$ has at least $r+1$ vertices (otherwise, %$|V\setminus T| \leq r + r(k-1) = rk$).
%Any maximal independent set $T'$ of $G'[V\setminus T]$ has %$|T'| > r = |T|$ and $V \setminus T'$ is a $k$-dominating %set such that $|V \setminus T'| < |V \setminus T|$. 
%
%Hence, we can state the following constructive version of %the result of \cite{CocGamShe85}.}

\begin{theorem}\label{AlgokNB} For a given graph $G$ of order $n_G$ and minimum degree at least $k$, one can compute a $k$-dominating set $D$  with $|D|\leq \frac{k}{k+1}n_G$ in polynomial time.
\end{theorem}

% The algorithm proceeds by iteratively improving the size of an independent set until its complement is of size $\leq kn/(k+1)$.

%  While $|V \setminus T| \leq kn/(k+1)$ do 
%            construct an independent set $T'$ of $G[V \setminus T]$ and set $T=T'$.

\begin{proof}
% We describe a polynomial time algorithm in what follows. 
% Let $G=(V,E)$ be a graph of order $n$ with minimum degree at least $k$. 
% We are going to construct a $k$-dominating set $D$ for $G$ with $|D|\leq kn/(k+1)$.
First, we greedily remove edges between  vertices of degree greater than  $k$ obtaining a graph $G'$ of minimum degree (exactly) $k$. 
 Let $S = \{v\in V: d(v) > k\}$. By construction, $S$ is an independent set in $G'$.
We build a maximal independent set $T$ that contains $S$. Then $V\setminus T$ is a $k$-dominating set.

If $|V \setminus T| \leq kn_G/(k+1)$, then $D:=V\setminus T$ is also a $k$-dominating set in the supergraph $G$ of $G'$.
Otherwise, while $|V \setminus T| > kn_G/(k+1)$, construct a maximal independent set $T'$ of $G[V \setminus T]$ and set $T=T'$.
We show in the following that the algorithm terminates.

% Since every vertex of $V\setminus T$ is of degree $k$ in $G'$, it must have degree less than  $k$ in $G'[V\setminus T]$.

Let $r = |T|$. When $|V \setminus T| > kn_G/(k+1)$, we get $n_G=r+|V\setminus T| > r+ kn_G/(k+1)= r+ k(r+|V \setminus T|)/(k+1)$, thus $|V \setminus T| > kr$.
Since $T$ is a maximal independent set (and hence a dominating set) and every element of $V\setminus T$ is of degree $k$ in $G'$, every vertex of $V\setminus T$  has degree at most $k-1$ in $G'[V\setminus T]$. It follows that any maximal independent  set of $G'[V\setminus T]$ contains at least $r+1$ vertices (otherwise, $|V\setminus T| \leq r + r(k-1) = rk$).
Compute any maximal independent set $T'$ of $G'[V\setminus T]$.
%\marginpar{Above, you write  $G[V \setminus T]$. What do you mean??} 
Now $|T'| > r = |T|$ and $V \setminus T'$ is a $k$-dominating set that is smaller than $V \setminus T$; namely, because the minimum degree is at least $k$, all elements of any independent set are $k$-dominated by its complement. \qed
%\marginpar{Why is $V \setminus T'$ is a $k$-dominating set?} 
\end{proof}

% Our approximation algorithm is based on 
% obtaining a $k$-dominating set in a graph $H$ obtained from
% the input $G$ after a number of modifications\LV{ (mainly vertex deletions and insertions)}.

\begin{theorem}\label{KDS2kNB}
Let $G$ be a graph of order $n_G$ and let $H$ be a graph of order $n_H$ obtained from $G$ by deleting $d$ vertices and adding $2k$ new vertices,\LV{ with} $d > k$. Let $D_G$ and $D_H$ be $k$-dominating set solutions of $G$
and $H$ such that $a = |D_G|-|D_H| = d-k$. 
If $|D_H| \leq c\cdot n_H$ and $d\leq \gamma_k(G)$, then $V(G)\setminus D_G$ is a 
$k$-nonblocker of $G$ whose size  $n_G-|D_G|$ is within a factor of $(1-c)^{-1}$ from optimum (modulo an additive constant less than $k$).
\end{theorem}

\LV{
\begin{proof} 
As $n_H=n_G-d+2k$,
$|D_G| = |D_H|+a \leq c(n_G-d+2k)+d-k
\leq cn_G + (1-c)d +2ck-k \leq cn_G + (1-c)\gamma_k(G)+ k(2c-1)$.
Hence, $n_G-|D_G| \geq n_G - cn_G - (1-c) \gamma_k(G) - k(2c-1) = (1-c)(n_G-\gamma_k(G)) - k(2c-1)$.
This immediately yields an approximation factor of $(1-c)^{-1}$ (modulo the additive constant $k(2c-1)\leq \frac{k(k-1)}{k+1}< k$).\qed
\end{proof}}

\noindent
% This result is understood modulo the additive constant $k(2c-1)\leq \frac{k(k-1)}{k+1}\leq k$.
\SV{ Proofs of this and other results can be found in the long version of this paper.}

%Notice that Theorem~\ref{KDS2kNB} leads to a factor-$(k+1+\varepsilon)$ approximation algorithm for \textsc{$k$-Nonblocker} based on Theorem~\ref{AlgokNB}.

%\medskip

In the rest of this section, we present reduction rules that produce a graph $G$ with minimum degree at least $k$.
%
%\medskip
%
%%\noindent
%{\bf Reduction Rules for $k$-\textsc{Nonblocker}}\\
%
%\noindent
Our reduction rules mainly deal with vertices  of degree $k-1$ or less. Each such vertex must be in any $k$-dominating set. We shall refer to such vertices by {\em low-degree} vertices in the sequel.

\paragraph{\bf Low-Degree Vertex Deletion\LV{ Reduction}\SV{.}} If a low-degree vertex $v$ has only low-degree neighbors, then delete $v$. If there is a vertex $u$ with $k+1$ low-degree neighbors, then delete one neighbor of $u$.

\begin{observation}\LV{The} \label{low-degree deletion}
Low-Degree Vertex Deletion\LV{ Reduction} is $\alpha$-preserving for any $\alpha\geq 1$.
\end{observation}

\LV{
\begin{proof} 
The soundness of Low-Degree Vertex  Deletion is rather straightforward. A low-degree vertex that is not a neighbor of a high-degree vertex can be placed (safely) in any $k$-dominating set.
If the number of low-degree neighbors of a vertex $u$ is $t > k$, then we can safely delete $t-k$ such neighbors and place them in the $k$-dominating set. We keep $k$ neighbors to make sure any subsequent solution places $u$ in the nonblocker set. This reduction is  $\alpha$-preserving with constants $a=b=0$.\qed
\end{proof}}

\paragraph{\bf Low-Degree Merging\LV{ Reduction}\SV{.}} Let $G$ be an instance of \textsc{$k$-Nonblocker} that has been subject to the Low-Degree Vertex Deletion\LV{ Reduction} rule. Then we add a complete bipartite graph $K_{k,k}$ with new vertices $u_1,\ldots,u_k$, $v_1,\ldots,v_k$.  
For every high-degree vertex $v\in V$ having $q$ low-degree neighbors $w_1,\ldots,w_q$, with $q\leq k$, delete  $w_1,\ldots,w_q$, and connect $v$ to $v_1,\ldots,v_q$.

\begin{observation}
\label{Leaf Merging}\LV{The}
Low-Degree Merging\LV{ Reduction rule} is $\alpha$-preserving for any $\alpha\geq 1$.
\end{observation}

\LV{
\begin{proof} 
We are going to verify the definition of $\alpha$-preserving reductions; to this end, we  show  that $a=b=k$ works out in our case.
Let $G=(V,E)$ be the original graph and $G'=(V',E')$ the graph obtained from $G$ by deleting $w_1,\ldots,w_q$, and connect $v$ to $v_1,\ldots,v_q$.

\smallskip\noindent
(a) Let $C$ be a maximum $k$-nonblocker for $G$. $C$ does not contain
$w_1,\ldots,w_q$ since $w_1,\ldots,w_q$ are part of any $k$-dominating set. Consider $C'=C \cup
\{u_1,\ldots,u_k\}$. $C'$ is a maximum $k$-nonblocker set for $G'$  of size $|C'|=|C|+k$.

\smallskip\noindent
(b) Consider now the converse. Let $C'$  be some $k$-nonblocker set for $G'$. We can suppose that $C'$ contains $u_1,\ldots,u_k$, otherwise we remove $v_1,\ldots,v_k$ and add $u_1,\ldots,u_k$. Thus $C= C'\setminus\{u_1,\ldots,u_k\}$ is a $k$-nonblocker set for $G$ of size $|C|=|C'|-k$.\qed
\end{proof}
}
%\noindent
The  reductions above take polynomial time, so that  Theorem~\ref{KDS2kNB} allows us to conclude:

\begin{theorem}For any instance $x$ of \textsc{$k$-Nonblocker}, one can compute in 
polynomial time a $k$-nonblocker set $S$ with $|S| \geq 
\frac{m^*(x)}{k+1}-k$.
%\textsc{$k$-Nonblocker} can be approximated up to a %factor of $k+1+\varepsilon$ in polynomial time.
%\marginpar{EPSI??}
\end{theorem}

\LV{Combinations with the previous section as indicated in the definitions of~\cite{HanVol2009} should be possible.
We leave this for future research, similar to variants like \textsc{Liar's Domination}; see \cite{BisGhoPau2013} and the  literature quoted therein.}

\section{Conclusions}

We presented  a framework for obtaining approximation\LV{ algorithm}s for maximization problems, inspired by similar reasonings for obtaining kernelization results. We see five \LV{major }directions from this approach: %\LV{
\begin{itemize}
\item%}\SV{(a)} 
Paraphrasing \cite{Estetal2005}, we might say that not only FPT, but also polynomial-time maximization is \emph{{P}-time extremal structure}.
This should inspire mathematicians working in\LV{ graph theory (and other areas of} combinatorics\LV{)} to work out useful \LV{combinatorial }bounds on different graph parameters. We started on domination-type parameters, and this might be a first venue of continuation, for example, along the lines sketched in \LV{\cite{BorMic98,BouBliChe2014,KamVol2009}}\SV{\cite{BorMic98,KamVol2009}}.
%\LV{
\item%}\SV{(b)} 
Conversely, approximation algorithms that stay within the combinatorial grounds of their problem tend to reveal (combinatorial) insights into the problem that might get lost when moving for instance into the area of Mathematical Programming. 
%\LV{
\item%}\SV{(c)} 
The notion of $\alpha$-preserving reduction is similar to the local ratio techniques~\cite{Baretal2004} that allowed to re-interpret many \LV{(e.g., primal-dual) }approximation algorithms (for minimization problems) in a purely combinatorial fashion\LV{; see \cite{BarRaw2005}}.
We see \LV{some }hope for similar developments using $\alpha$-preserving reduction for maximization problems.
%\LV{
\item%}\SV{(d)} 
The fact that  $\alpha$-preserving reductions are inspired by FPT techniques should allow to adapt these notions for\LV{ obtaining new and faster} parameterized approximation algorithms. %\LV{
\item%}\SV{(e)} 
Reductions are often close to practical heuristics and hence allow for fast implementations.
%\LV{
\end{itemize}%}

\SV{\small\smallskip\noindent\textit{Acknowledgements.}} 
\LV{\paragraph{Acknowledgements.}}
We are grateful for\LV{ discussing this paper at the Bertinoro Workshop on Parameterized Approximation in May 2014}\SV{ the support by the bilateral research cooperation
CEDRE between France and  Lebanon (grant number
30885TM)}.
\SV{\normalsize}

% \LV{\bibliography{abbrev,hen,unpubhen,zukopieren}}
% \SV{\bibliography{ISAAC}}

\end{document}